\documentclass[a4paper,UKenglish, thm-restate]{lipics-v2021}
\usepackage{hyperref,cleveref}
\usepackage{amsmath,amssymb,amsthm}
\usepackage{paralist}
\usepackage{xcolor}
\usepackage{xspace}
\usepackage{bbm}
\usepackage{multirow}

\usepackage{tablefootnote}

\usepackage{makecell}

\usepackage{graphicx} 

\usepackage{algorithm}
\usepackage[noend]{algpseudocode}

\newcommand{\ts}{\textsuperscript}
\newcommand{\jacobus}[1]{\textcolor{green!50!blue}{Jacobus: #1}}
\newcommand{\andre}[1]{\textcolor{black!10!red}{André: #1}}

\hypersetup{
  colorlinks=true,
  linkcolor=red!70!black,
  linktoc=all,
  citecolor=blue,
    pdfpagemode=FullScreen,
}

\definecolor{Purple}{cmyk}{0.45,0.86,0,0}
\renewcommand{\emph}[1]{{\color{Purple}{\em #1}}}

\def\macrocolor{.}
\newcommand{\mathmacro}[1]{{\ensuremath{\textcolor{\macrocolor}{#1}}}}

\def\bR{\mathmacro{\mathbb R}}
\def\tildeO{\mathmacro{\tilde{\mathcal{O}}}}
\def\O{\mathmacro{\mathcal{O}}}
\def\area{\mathmacro{\mathrm{area}}}
\def\eps{\mathmacro{\mathrm{\varepsilon}}}

\def\rectPolyProb{{\textsc{OrthMaxOverlap}}\xspace}
\def\polyCont{{\textsc{PolygonContainment}}\xspace}
\def\aindex{\mathmacro{h}}
\def\bindex{\mathmacro{i}}
\def\cindex{\mathmacro{j}}
\def\dindex{\mathmacro{k}}
\def\eindex{\mathmacro{\ell}}

\newtheorem{thm}{Theorem}

\newtheorem{problem}[thm]{Problem}

\title{Algorithms and Lower Bounds for the Maximum Overlap of Two Polygons Under Translation}
\titlerunning{On the Maximum Overlap of Two Polygons Under Translation}

\author{Mikkel Abrahamsen}{University of Copenhagen,  Denmark\and \url{https://sites.google.com/view/mikkel-abrahamsen}}{miab@di.ku.dk}{https://orcid.org/0000-0003-2734-4690}{Supported by Independent Research Fund Denmark, grant 1054-00032B, and by the Carlsberg Foundation, grant CF24-1929.}
\author{Sujoy Bhore}{Indian Institute of Technology Bombay, India\and \url{https://sites.google.com/view/homepage-of-sujoy-bhore/home}}{sujoy@cse.iitb.ac.in}{https://orcid.org/0000-0003-0104-1659}{}
\author{Maike Buchin}{Ruhr University Bochum, Germany}{maike.buchin@rub.de}{https://orcid.org/0000-0002-3446-4343}{}
\author{Jacobus Conradi}{University of Bonn, Germany \and University of Copenhagen, Denmark \and \url{https://jacobus-conradi.eu} }{jacobus.conradi@gmx.de}{https://orcid.org/0000-0002-8259-1187}{Supported partially by the iBehave Network: Sponsored by the Ministry of Culture and Science of the State of North Rhine-Westphalia, and the Carlsberg Foundation, grant CF24-1929.}
\author{Ce Jin}{Massachusetts Institute of Technology, USA \and University of California, Berkeley, USA}{cejin@berkeley.edu}{https://orcid.org/0000-0001-5264-1772}{Supported by the Jane Street Graduate Research Fellowship, NSF grant CCF-2330048, and a Simons Investigator Award.}
\author{André Nusser}{Université Côte d'Azur, CNRS, Inria, France}{andre.nusser@cnrs.fr}{https://orcid.org/0000-0002-6349-869X}{Supported by the France 2030 investment plan managed by the ANR as part of the Initiative of Excellence of Université Côte d'Azur with reference number ANR-15-IDEX-01.}
\author{Carolin Rehs}{Technical University of Dortmund, Germany \and Technical University of Eindhoven, The Netherlands}{carolin.rehs@tu-dortmund.de}{https://orcid.org/0000-0002-8788-1028}{Partly supported by the PRIME programme of the German Academic Exchange Service (DAAD) with funds from the German Federal Ministry of Research, Technology and Space (BMFTR)}

\authorrunning{M. Abrahamsen, S. Bhore, M. Buchin, J. Conradi, C. Jin, A. Nusser and C. Rehs}
\Copyright{M. Abrahamsen, S. Bhore, M. Buchin, J. Conradi, C. Jin, A. Nusser and C. Rehs}
\date{February 2025}


\begin{document}
\nolinenumbers

\maketitle
\begin{abstract}
\noindent Given two polygons of complexities $n$ and $m$ respectively, a fundamental problem in shape matching and geometric similarity is to compute their maximum area overlap under translation.
For general simple polygons, the best-known algorithm runs in $\mathcal{O}((nm)^2 \log(nm))$ time [Mount, Silverman, Wu ’96].
In a recent breakthrough that received the SoCG Best Paper Award 2025, Chan and Hair gave a linear-time algorithm for the special case when both polygons are convex. A key challenge in computational geometry is to design improved algorithms for other natural classes of polygons. We address this by presenting an $\mathcal{O}((nm)^{3/2} \log(nm))$-time algorithm for the case when both polygons are orthogonal, probably the most popular class of polygons besides convex and simple ones.
This is the first algorithm for polygon overlap on orthogonal polygons that is faster than the almost 30 years old algorithm for general simple polygons.

Complementing our algorithmic contribution, we provide $k$-SUM lower bounds for problems on simple polygons with only orthogonal and diagonal edges. 
First, we establish that there is no algorithm for polygon overlap with running time $\mathcal{O}(\max(n^2,nm^2)^{1-\varepsilon})$, where $m\leq n$, unless the $k$-SUM Hypothesis fails.
This matches the running time of our algorithm when $n=m$.
We use part of the above construction to also show a lower bound for the polygon containment problem, a popular special case of the overlap problem. Concretely, there is no algorithm for polygon containment with running time $\mathcal{O}(n^{2-\varepsilon})$ under the $3$-SUM Hypothesis, even when the polygon to be contained has $m \in \mathcal{O}(1)$ vertices.
Our lower bound shows that polygon containment for these types of polygons (i.e., with diagonal edges) is strictly harder than for orthogonal polygons, and also strengthens the previously known $3$-SUM lower bound for polygon containment of [Barequet, Har-Peled~'01].
Furthermore, our lower bounds show conditional tightness (up to polylogarithmic factors) of the algorithms of [Avnaim, Boissonnat~'89] and [Mount, Silverman, Wu~’96] when $m \in \mathcal{O}(1)$.
\end{abstract}

\newpage

\section{Introduction}\label{intro}


We consider the problem of computing the \emph{maximum overlap} of two polygons $P$ and $Q$ with $n$ and $m$ vertices, respectively, i.e., finding a translation $\tau = (\tau_x, \tau_y) \in\mathbb R^2$ such that the area of $P\cap (Q+\tau)$ is maximum.
This can be seen as a formalization of matching the shapes $P$ and $Q$, and the maximum overlap is a natural similarity measure: the larger it is relative to the areas of $P$ and $Q$, the more similar they are. Shape matching itself is a core technique in geographic information science, computer aided design, robotics, and medical imaging. Applications include remote sensing, fingerprint matching, molecular shape matching, and image sequence alignment \cite{meng2012fast,choi2007fingerprint,kim2011pubchem3d,10.1007/3-540-55426-2_53}.

The maximum overlap problem for general simple polygons was first studied by Mount, Silverman and Wu~\cite{DBLP:journals/cviu/MountSW96}.
They show that the overlap function $\tau\mapsto \area (P\cap (Q+\tau))$ is a continuous and piecewise polynomial surface of degree at most two, consisting of $\O((nm)^2)$ pieces.
Furthermore, they give an example of two orthogonal polygons (for arbitrary values of $n$ and $m$) that realize the bound, i.e., the complexity of the overlap function is $\Omega((nm)^2)$, and hence the bound is tight in the worst-case.
They also describe an algorithm to compute the entire function within the same bound $\O((nm)^2)$ on the running time, which yields an $\O((nm)^2 \log(nm))$ algorithm for the polygon overlap problem. They posed the open question whether an algorithm with better running time exists.
Approximation algorithms for finding the maximum overlap are described in~\cite{DBLP:journals/dcg/CheongEH07, DBLP:journals/comgeo/ChengL13,DBLP:journals/algorithmica/Har-PeledR17}. See also the survey~\cite{o2017polygons} on polygons from the handbook of discrete and computational geometry for an overview. 

A prominent special case is when both polygons are convex, for which the bound $\O(m^2 +n^2 + mn \cdot \min(m, n))$ on the complexity of the overlap function was given in~\cite{DBLP:journals/cviu/MountSW96}.
De Berg, Cheong, Devillers, van Kreveld, and Teillaud~\cite{DBLP:journals/mst/BergCDKT98} describe an algorithm that computes the maximum overlap of two convex polygons in time $\tildeO(n +m)$ ($\tildeO(\cdot)$ hides logarithmic factors in $n$ and $m$).
This was improved recently by Chan and Hair in a celebrated paper~\cite{DBLP:conf/compgeom/ChanH25} that received the SoCG Best Paper Award 2025. They describe a randomized algorithm with $\O(n+m)$ expected running time.

Maximizing the overlap has also been studied for convex polyhedra in higher dimensions, with Ahn, Brass, and Shin \cite{ahn2008maximum} providing an algorithm that runs in $\tildeO(n^{d+1-3/d})$ expected time. This was later improved upon by Ahn, Cheng and Reinbacher \cite{ahn2013maximum}, who presented an algorithm that for any $\mu>0$ computes a translation for which the area of the overlap is at most the optimum overlap minus $\mu$. The algorithm has running time $\tildeO(n)$ for $d=3$, and $\tildeO(n^{\lfloor d/2\rfloor+1})$ for $d\geq 4$.

Another variant of the problem that has been studied extensively is the minimization of the area of the convex hull of $P\cup (Q+\tau)$. For this problem, Jung, Kang, and Ahn \cite{jung2025minimum} presented an algorithm with deterministic running time $O(n)$ for $d=2$, $O(n^{(d+1)/2})$ for odd $d\geq 3$, and $\tildeO(n^{d/2})$ for even $d\geq 3$.

\begin{table}
    \caption{Upper and lower bounds for problems related to \rectPolyProb. 
    }
    \label{tab:bounds}
    \centering
    \renewcommand{\arraystretch}{1.2}
    \begin{tabular}{|c|c|c|c|}
        \hline
        Polygon Type & Objective & Lower Bound & Upper Bound \\
        \hline
        \hline
        \multirow{2}{*}{\textsc{Convex}} & \textsc{Containment} & $\Omega(n+m)$ & $\O(n+m)$ \cite{DBLP:conf/compgeom/ChanH25}\\\cline{2-4}
        & \textsc{MaxOverlap} & $\Omega(n+m)$ & $\O(n+m)$ \cite{DBLP:conf/compgeom/ChanH25}\\
        \hline
        \multirow{3}{*}{\textsc{Orthogonal}} & \textsc{Containment} & $\O((nm)^{1-\eps})$ \cite{KunnemannPolygonPlacement2021} &$\tildeO(nm)$ \cite{DBLP:journals/ipl/Barrera96}\\\cline{2-4}
        & \multirow{2}{*}{\textsc{MaxOverlap}} & \multirow{2}{*}{$\O((nm)^{1-\eps})$ \cite{KunnemannPolygonPlacement2021}} & $\tildeO((nm)^{3/2})$ \\
        & & & \Cref{thm:alg} \\
        \hline
        \multirow{5}{*}{\makecell{\textsc{Octolinear}\footnotemark \\ or \\ \textsc{Simple}}} & \multirow{2}{*}{\textsc{Containment}} & $\O(n^{2-\eps})$ & \multirow{2}{*}{$\tildeO((nm)^2)$ \cite{DBLP:journals/ita/AvnaimB89}}\\
        & & \Cref{cor:3sum} &\\\cline{2-4}
        & \multirow{3}{*}{\textsc{MaxOverlap}} & $\O(\max(n^2,nm^2)^{1-\eps})$ & \multirow{3}{*}{$\O((nm)^2)$ \cite{DBLP:journals/ita/AvnaimB89}}\\
        & & \Cref{thm:5sum} &\\
        & & $m=\Theta(n)$: $\O((nm)^{(3/2)-\eps})$ &\\
        \hline
    \end{tabular}
\end{table}
\footnotetext{The class of \textit{octolinear} polygons is the subclass of simple polygons whose edges are at an angle $k\frac{\pi}{4}$ for $0\leq k< 8$. We remark that the lower bound requires only six out of all eight edge angles, and the class of orthogonal polygons is trivially contained in the class of octolinear polygons.}

In this paper, we consider the maximum overlap problem in $\bR^2$ for a different well-studied subclass of simple polygons, 
where both $P$ and $Q$ are orthogonal polygons, i.e., simple polygons where every edge is either horizontal or vertical:

\begin{problem}[\rectPolyProb]
    Let $P$ and $Q$ be two orthogonal polygons with $n$ and $m$ vertices respectively. Find $\tau\in \bR^2$ such that $\area(P\cap (Q+\tau))$ is maximum.
\end{problem}

For this problem, we give the first improvement over the nearly three decades old $\O((nm)^2 \log(nm))$ bound from~\cite{DBLP:journals/cviu/MountSW96}. In particular, we prove the following theorem.

\begin{restatable}{theorem}{fastAlg}\label{thm:alg}
    
Let $P$ and $Q$ be orthogonal polygons with $n$ and $m$ vertices, respectively.
There is an algorithm for computing the translation $\tau$ maximizing the overlap $\area(P\cap (Q+\tau))$ in time $\O((nm)^{3/2}\log(nm))$.
\end{restatable}

We complement our algorithm by a conditional lower bound via the $k$-SUM problem.

\begin{problem}[$k$-SUM]
    Given $k$ sets of positive integers $A_1$, $A_2$, $\ldots$, $A_k$ of size $n$ each, is there $(a_1,\ldots,a_k)\in A_1 \times\ldots\times A_k$ such that $a_1=\sum_{i\geq 2}^k a_i$?
\end{problem}

The corresponding hardness assumption in fine-grained complexity theory is the $k$-SUM-Hypothesis \cite{abboud2013exact, GAJENTAAN1995165}:

\begin{definition}[$k$-SUM Hypothesis]
    For any integer $k$ the $k$-SUM problem cannot be solved in $\O(n^{\lceil k/2\rceil-\eps})$ time for any $\eps>0$.    
\end{definition}

We provide the following lower bound for any algorithm solving the maximum overlap problem for simple polygons, assuming the $k$-SUM Hypothesis.

\begin{restatable}{theorem}{thmFiveSum}\label{thm:5sum}
Let $P$ and $Q$ be polygons with $n$ and $m$ vertices, respectively, where $m\leq n$ and $\delta \in \mathbb{R}^+$. 
Assume that all edges of $P$ and $Q$ are either horizontal, vertical or parallel to the diagonal.
Assuming the $k$-SUM Hypothesis, deciding whether there is a translation $\tau$ such that $\area(P\cap(Q+\tau))>\delta$ cannot be solved in $\O((\max(n^2,nm^2))^{1-\eps})$ time for any $\eps >0$.
\end{restatable}


Our reduction can also be modified to obtain a lower bound for the \emph{polygon containment} problem. While this problem is a special case of the \emph{maximum overlap} problem, it has attracted more attention than the overlap problem.
In the containment problem we want to know whether 
there is some $\tau\in\mathbb R^2$ so that $(Q+\tau)\subset P$ and thus
$\area (P\cap (Q+\tau))=\area (Q)$.
This problem was first studied by Chazelle in 1983~\cite{chazelle1983polygon}, followed by numerous other results. 
Of particular interest, we only mention the results when both $P$ and $Q$ are simple or orthogonal polygons and only translations are allowed; a plethora of results exist when rotations are allowed and for more restricted variants such as when one or both of the polygons are convex; for exhaustive literature reviews on the problem, see the recent papers~\cite{KunnemannPolygonPlacement2021,DBLP:conf/compgeom/ChanH24}.
Avnaim and Boissonnat~\cite{DBLP:journals/ita/AvnaimB89} gave an algorithm for solving polygon containment of two simple polygons in $\O((nm)^2\log(nm))$ time.
Hern{\'{a}}ndez{-}Barrera~\cite{DBLP:journals/ipl/Barrera96} improved the running time to $\O((nm)^2)$ and solved the case where both $P$ and $Q$ are orthogonal in $\O(nm\log (nm))$ time.

As for known lower bounds, Barequet and Har{-}Peled~\cite{DBLP:journals/ijcga/BarequetHardness01} showed that when $P$ and $Q$ both have $n$ vertices, no algorithm has running time $\O(n^{2-\varepsilon})$ under the $3$-SUM Hypothesis.
K{\"{u}}nnemann and Nusser~\cite{KunnemannPolygonPlacement2021} gave various conditional lower bounds for polygon containment and the more general \emph{largest scaled containment} problem, which is to find the largest homothet (i.e., scaled copy) of $Q$ that fits in $P$.
For polygon containment, they showed that when $P$ and $Q$ are orthogonal and have arbitrary numbers $n$ and $m$ of vertices, then no algorithm has running time $\O((nm)^{1-\varepsilon})$ time under the Orthogonal Vectors Hypothesis.
Assuming the $4$-SUM Hypothesis, they showed that for the largest scaled containment problem, no algorithm has running time $\O(n^{2-\varepsilon})$, even when both $P$ and $Q$ are orthogonal and $m=O(1)$.

We improve on these lower bounds by simplifying the reduction behind \Cref{thm:5sum} to show the following.  

\begin{restatable}{theorem}{polyContLowerBound}\label{cor:3sum}
Let $P$ and $Q$ be polygons with $n$ and $\O(1)$ vertices, respectively.
Assume that all edges of $P$ and $Q$ are either horizontal, vertical or parallel to the diagonal.
Assuming the $3$-SUM Hypothesis, deciding whether there is a translation for $Q$ such that $Q$ is contained in $P$ cannot be solved in $\O(n^{2-\varepsilon})$ time for any $\eps>0$.
\end{restatable}


We summarize all our results in context in Table~\ref{tab:bounds}, and highlight below some key consequences of our theorems:


\begin{itemize}
    \item The arrangement that is constructed for the polygon overlap problem in \cite{DBLP:journals/cviu/MountSW96} has worst-case size $\Theta((nm)^2)$, even for orthogonal polygons. We break this natural barrier with our $\O((nm)^{3/2} \log nm)$ running time.
    \item Besides improving the known lower bounds for polygon containment for any $m = \O(n^\alpha)$ with $\alpha<1$, we give the first non-trivial lower bound for polygon containment in the regime where $m \in \O(1)$, improving the previous lower bounds from $\O(n^{1-\eps})$ to $\O(n^{2-\eps})$.
    \item We show that polygon containment with edges in three directions (orthogonal and one diagonal) is strictly harder than polygon containment with edges in two directions (orthogonal): we prove that the former cannot be solved in $\O(n^{2-\eps})$ time when $m \in O(1)$, while the latter is solvable in near-linear time in this setting.
    \item With our improved lower bounds, we show that the decades-old algorithms of \cite{DBLP:journals/ita/AvnaimB89, DBLP:journals/cviu/MountSW96, DBLP:journals/ipl/Barrera96} for polygon overlap and containment are all tight up to polylogarithmic factors for $m \in \O(1)$.
\end{itemize}

\subsection{Technical Overview}

We first give an overview of our techniques used to obtain our results.

\subparagraph{Upper bound}

Let $P$ and $Q$ be two orthogonal polygons with $n$ and $m$ vertices respectively. We give the first algorithm improving on the original $\O((nm)^2)$ algorithm from Mount, Silverman and Wu \cite{DBLP:journals/cviu/MountSW96} for non-convex polygons. That is, for the subclass of orthogonal polygons we give an $\tildeO((nm)^{3/2})$ algorithm.

Similar to the approach presented in \cite{DBLP:journals/cviu/MountSW96}, we consider the piecewise polynomial surface of degree at most two describing the function $\tau\mapsto\area(P\cap(Q+\tau))$. We observe that the function splits into $\O(nm)$ rectangles, each with their own bilinear function, and $\tau\mapsto\area(P\cap(Q+\tau))$ can be expressed as the sum of bilinear functions of rectangles containing $\tau$. We may now sweep over these rectangles to compute the translation maximizing the intersection of $P$ and $Q+\tau$ in $\tildeO((nm)^2)$ time. We improve this sweep-line algorithm with a batching technique which groups events for the sweep-line algorithm into $\O(\sqrt{nm})$ groups of size $\O(\sqrt{nm})$ in both $x$- and $y$-direction resulting in $\O(nm)$ small windows in which we identify the maximizing translation in $\tildeO(\sqrt{nm})$ time as follows: Via algebraic manipulations we rewrite the partial sums of bilinear functions, as they would have been encountered and evaluated by the sweep-line algorithm, into a set $V$ of $\O(\sqrt{nm})$ vectors in $\bR^3$. Identifying the maximum among the $\O(nm)$ candidate translations in the small window then reduces to $\O(\sqrt{nm})$ extreme-point queries on the set $V$ which, as $V\subset\bR^3$, can be answered in $\tildeO(\sqrt{nm})$ total time, via the convex hull of $V$. The result is an algorithm with running time in $\tildeO((nm)^{3/2})$. 

Comparing our techniques to the algorithm given by Mount, Silverman and Wu~\cite{DBLP:journals/cviu/MountSW96}, we observe that our approach is similar in spirit in that we consider the piecewise polynomial surface of degree at most two describing the function $\tau\mapsto\area(P\cap(Q+\tau))$. 
As noted in Section~\ref{intro}, Mount, Silverman and Wu showed that the overlap function can have complexity $\Omega((nm)^2)$ even for orthogonal polygons, so their method of explicitly computing the polynomial surface is unlikely to yield improved bounds.
Unlike their algorithm, however, we manage to circumvent the explicit computation of the polynomial surface represented as the sum of bilinear polynomials defined on rectangular regions via our novel recast to extreme-point queries to a three-dimensional convex hull. 
With this, we improve the running time of the algorithm given by Mount, Silverman and Wu by a factor of $\sqrt{nm}$.

Overall, our techniques reduce a geometric problem to an algebraic problem, which we recast into a geometric problem to solve efficiently. Interestingly, there is very little conceptual overlap between the initial geometric problem and the final recast, and little intuition can be gathered from the set $V$ for the underlying \rectPolyProb instance. 

\subparagraph{Lower bounds} 
We complement our algorithm by giving lower bounds for the polygon overlap and for the polygon containment problem. We reduce to the polygon overlap problem from an imbalanced variant of $5$-SUM. In this variant, we are given five sets $A$, $B$, $C$, $D$, and $E$ with $|A| = |B| = |C| = n$ and $|D| = |E| = m$,
and we want to decide whether there exist elements $(a,b,c,d,e)\in A\times B\times C\times D\times E$ such that $a=b+c+d+e$.
We show that this problem cannot be solved in time $\O(\max(n^2, nm^2)^{1-\eps})$ unless the $k$-SUM Hypothesis fails.

Let us now sketch our reduction from this problem to the maximum overlap problem.
We observe that, for two comb-like polygons with their prongs oriented downward, the translation $\tau = (\tau_x, \tau_y)$ maximizing the overlap between both polygons aligns one prong from the translated polygon $Q$ with one prong of the fixed polygon $P$.
If the two aligned prongs are at $x$-coordinate $-d$ and $b$, then the $x$-coordinate of the translation satisfies the equation $\tau_x = b+d$.
Hence, by carefully spacing the prongs of the combs and creating a similar construction for the $y$-coordinate, we can restrict any translation aligning both pairs of combs at once to $\tau_x= b_{\bindex}+d_{\dindex}$ and $\tau_y = c_{\cindex}+e_{\eindex}$ for any $b_{\bindex}\in B$, $c_{\cindex}\in C$, $d_{\dindex}\in D$ and $e_{\eindex}\in E$. 
We now add a diagonal comb to the fixed polygon $P$, where every prong is placed at the (anti-)diagonal line $x+y=a_{\aindex}$ for $a_{\aindex}\in A$, and a single prong to $Q$ at $x+y=0$. 
These prongs are only aligned for translations of the form $\tau_x+\tau_y=a_{\aindex}$.
Then, if $\tau$ is a translation such that all three pairs of combs are aligned properly, we obtain $\tau_x= b_{\bindex}+d_{\dindex}$, $\tau_y = c_{\cindex}+e_{\eindex}$, and $\tau_x+\tau_y=a_{\aindex}$, which implies a solution to the $5$-SUM instance.
The polygons $P$ and $Q$ have $\O(n)$ and $\O(m)$ vertices respectively.
When each comb of $Q$ only has one prong, the construction extends to polygon containment, giving a lower bound based on $3$-SUM.

Comparing our techniques with the techniques of the 3-SUM Hypothesis based lower bound from Barequet and Har{-}Peled~\cite{DBLP:journals/ijcga/BarequetHardness01} and the Orthogonal Vectors Hypothesis based lower bound from K{\"{u}}nnemann and Nusser~\cite{KunnemannPolygonPlacement2021} for the \rectPolyProb problem, we observe the following: Both their techniques essentially consider only one pair of combs, 
restricting the set $\tau=(\tau_x,\tau_y)$ of translations such that $Q+\tau$ is contained in $P$ to one of $nm$ possible $x$-translations. 
Combining this with a second set of combs restricting the $y$-translations is straight-forward, however, there is no interdependence between $\tau_x$ and $\tau_y$, as both combs can be aligned independently of one another, 
giving $\O(nm)$ instead of the desired $\O((nm)^2)$ translations that `need to be checked'. This problem is addressed in the 4-SUM lower bound of~\cite{KunnemannPolygonPlacement2021} by considering scalings additionally to translations, to determine the biggest copy of $Q$ in $P$.
This additional degree of freedom allows to encode more complexity, however, at the cost of showing a lower bound for a more general problem (hence, a weaker bound).
In our work, we address this issue by introducing diagonal combs which encode equations of the form $\tau_x+\tau_y=C$ for some constant $C$. That allows us to couple $\tau_x$ and $\tau_y$ to one another such that the two axis-aligned combs \textit{cannot} be aligned independently from each other.
A somewhat similar technique of encoding addition of a 3-SUM problem via diagonal translations was used by Bringmann and Nusser in~\cite{DBLP:journals/jocg/BringmannN21}.

\section{Rectangles and Unions of Rectangles}

We first discuss a simple sweep-line based algorithm which solves \rectPolyProb in $\tildeO((nm)^2)$ total time, by partitioning both polygons into a set of rectangles and restricting the solution space to $\O((nm)^2)$ candidate translations. This sweep-line based algorithm serves as the basis for an improved algorithm which we present in \Cref{sec:fast}. We begin by analyzing the function $\tau\mapsto\area(P\cap(Q+\tau))$, where $P$ and $Q$ are rectangles.

For this, define for two sets $A$ and $B$ in $\bR$, the set $A-B:=\{a-b \mid (a,b)\in A\times B\}$. Let further $\lambda(\cdot)$ be the Lebesgue measure, which we use to measure the length of intervals in $\bR$. 

\begin{figure}
    \centering
    \includegraphics[width=0.85\linewidth]{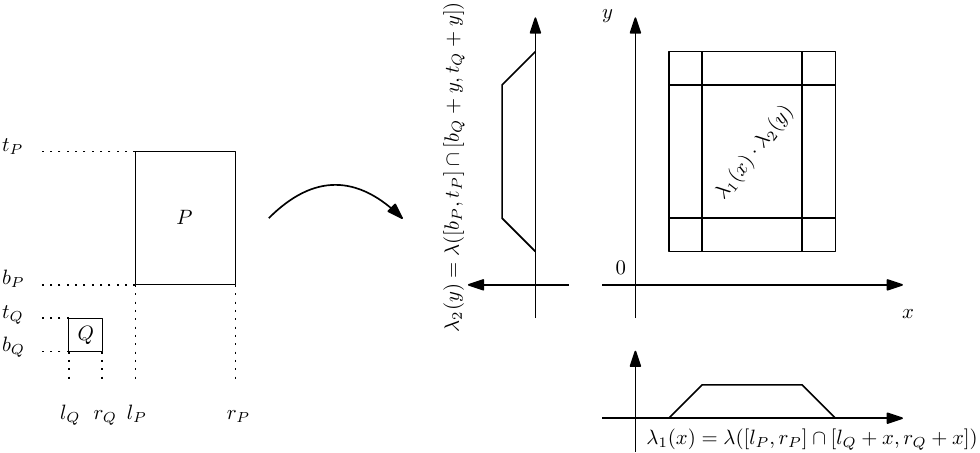}
    \caption{Illustration of the support of the function $\area(P\cap(Q+(x,y))$, and how it partitions into $\O(1)$ (half-open disjoint) rectangular regions.}
    \label{fig:rectrect}
\end{figure}

\begin{lemma}\label{lem:translationfunction}
    Let $P=[l_P,r_P]\times[b_P,t_P]$ and $Q=[l_Q,r_Q]\times[b_Q,t_Q]$ be two rectangles. There are $9$ disjoint half-open rectangles $[l_i,r_i)\times[b_i,t_i)$ and real values $A_i,B_i,C_i$ and $D_i$ computable in $\O(1)$ time such that $l_i,r_i\in \{l_P,r_P\}-\{l_Q,r_Q\}$ and $b_i,t_i\in \{b_P,t_P\}-\{b_Q,t_Q\}$ and
    \[\area(P\cap (Q+(x,y)))=
    \begin{cases}
        A_i + B_i x + C_iy + D_i xy, & \text{if $(x,y)\in [l_i,r_i)\times[b_i,t_i)$ for $i\leq 9$}\\
        0, & \text{otherwise}.
    \end{cases}\]
\end{lemma}
\begin{proof}
    For the proof refer to \Cref{fig:rectrect}.
    Observe that
    \[\area(P\cap (Q+(x,y))) = \lambda([l_P,r_P]\cap [l_Q+x,r_Q+x])\cdot \lambda([b_P,t_P]\cap [b_Q+y,t_Q+y]).\]
    The continuous function $x\mapsto\lambda([l_P,r_P]\cap [l_Q+x,r_Q+x])$ is a convolution of two one-dimensional interval membership functions and as such can be described by (at most) three disjoint half-open intervals $[l_i,r_i)$ and values $A_i$ and $B_i$ such that
    \[\lambda([l_P,r_P]\cap [l_Q+x,r_Q+x])=
    \begin{cases}
        A_i + B_i x , & \text{if $x\in [l_i,r_i)$ for $i\leq 3$}\\
        0, & \text{otherwise},
    \end{cases}\]
    where $l_i,r_i\in \{l_P,r_P\}-\{l_Q,r_Q\}$.
    Similarly there are (at most) three disjoint half-open intervals $[b_i,t_i)$ and values $C_i$ and $D_i$ such that 
        \[\lambda([b_P,t_P]\cap [b_Q+y,t_Q+y])=
    \begin{cases}
        C_i + D_i y , & \text{if $x\in [b_i,t_i)$ for $i\leq 3$}\\
        0, & \text{otherwise},
    \end{cases}\]
    with $b_i,t_i\in \{b_P,t_P\}-\{b_Q,t_Q\}$.
    Hence, the fact that $(A_i + B_ix)\cdot(C_j + D_jy) = A_iC_j + B_iC_jx + A_iD_jy + B_iD_jxy$ for $i,j\leq 3$ implies the claim. Finally, all values can be computed in $\O(1)$ time concluding the proof.
\end{proof}

\begin{corollary}\label{cor:translationslabs}
    Let $P=[l_P,r_P]\times[b_P,t_P]$ and $Q=[l_Q,r_Q]\times[b_Q,t_Q]$ be two rectangles. There are $18$ rectangular half-open regions $[l_i,r_i)\times[b_i,\infty)$ together with real values $A_i,B_i,C_i$ and $D_i$ computable in $\O(1)$ time such that
    \[\area(P\cap (Q+(x,y)))=\sum_i\mathbbm{1}[(x,y)\in[l_i,r_i)\times[b_i,\infty)]\left(A_i + B_i x + C_iy + D_i xy\right)\]
    Further, $l_i,r_i\in \{l_P,r_P\}-\{l_Q,r_Q\}$ and $b_i\in \{b_P,t_P\}-\{b_Q,t_Q\}$.
\end{corollary}
\begin{proof}
    This is an immediate consequence of \Cref{lem:translationfunction} together with the fact that for any $l_i$, $r_i$, $b_i$, $t_i$, $A_i$, $B_i$, $C_i$, $D_i$, $x$ and $y$ it holds that
    \begin{align*}
        &\mathbbm{1}[(x,y)\in[l_i,r_i)\times[b_i,t_i)](A_i + B_i x + C_iy + D_i xy)\\
        =&\mathbbm{1}[(x,y)\in[l_i,r_i)\times[b_i,\infty)](A_i + B_i x + C_iy + D_i xy)\\
        &-\mathbbm{1}[(x,y)\in[l_i,r_i)\times[t_i,\infty)](A_i + B_i x + C_iy + D_i xy).
    \end{align*}
\end{proof}

Next we show that the maximum of $\tau\mapsto\area(P\cap(Q+\tau))$ for orthogonal polygons $P$ and $Q$ is attained by one of only $\O((nm)^2)$ easily computable candidate translations.

\begin{lemma}\label{lem:polygonCutting}
    Let $P$ be an orthogonal polygon with $n$ vertices. In $\O(n\log n)$ time one can compute $\O(n)$ interior-disjoint rectangles $R_i$ such that $P=\bigcup_i R_i$. Further the coordinates of $R_i$ are a subset of the coordinates of $P$.
\end{lemma}
\begin{proof}
    This can be achieved via standard techniques, first sorting all the input coordinates according to their $y$-coordinate and then shooting horizontal rays from every reflex vertex (i.e., a vertex with interior angle $\frac{3\pi}{2})$, cutting the polygon into $\O(n)$ horizontal slabs, where every vertex defining such a slab is either an input vertex, or the intersection of a vertical edge of $P$ with a horizontal ray rooted at a reflex vertex. As such every point is comprised of coordinates of the input polygon. 
\end{proof}

\begin{lemma}\label{lem:discretization}
    Let $P$ and $Q$ be orthogonal polygons. Let $P_x,P_y,Q_x$ and $Q_y$ be the set of $x$- and $y$-coordinates of $P$ and $Q$ respectively. For $X=P_x-Q_x$ and $Y=P_y-Q_y$ it holds that
    \[\max_{\tau\in\bR^2}\area(P\cap(Q+\tau))=\max_{(x,y)\in X\times Y}\area(P\cap(Q+(x,y))).\]
\end{lemma}
\begin{proof}
    By \Cref{lem:translationfunction} and \Cref{lem:polygonCutting} there are $\O(nm)$ rectangular half-open disjoint regions $[l_i,r_i)\times[b_i,t_i)$ together with real values $A_i,B_i,C_i$ and $D_i$ such that for any $(x,y)\in\bR^2$
    \[\area(P\cap(Q+(x,y)))= \sum_{i}\mathbbm{1}[(x,y)\in [l_i,r_i)\times[b_i,t_i)]\cdot\left(A_i + B_i x + C_iy + D_i xy\right),\]
    with $l_i,r_i\in X$ and $b_i,t_i\in Y$. Let now $\tau^*=(x^*,y^*)$ be such that $\area(P\cap(Q+\tau^*))=\max_{\tau\in\bR^2}\area(P\cap(Q+\tau))$. Assume that $x^*\not\in X$. We will show that there is some $\hat{x}\in X$ such that for $\hat{\tau}=(\hat{x},y^*)$ it holds that $\area(P\cap(Q+\hat{\tau}))\geq \area(P\cap(Q+\tau^*))$. Note that this, together with a symmetric argument for the $y$-coordinate, implies the claim.

    As $x^*$ is not in $X$, there are two consecutive elements $x_l,x_r\in X$ with $x_l<x^*<x_r$.
    Let $\mathcal{R}$ be the set of all indices $i$ such that $\tau^*$ is in the half-open disjoint regions $[l_i,r_i)\times[b_i,t_i)$. 
    Thus for all $x\in[x_l,x_r)$ it holds that
    \begin{align*}
        \area(P\cap(Q+(x,y^*)))&= \sum_{i\in\mathcal{R}}\left(A_i + B_i x + C_iy^* + D_i xy^*\right)
    \end{align*}
    In fact, by continuity of $\tau\mapsto \area(P\cap(Q+\tau))$ this even holds for all $x\in [x_l,x_r]$. Let $B=\sum_{i\in\mathcal{R}}B_i$ and $D=\sum_{i\in\mathcal{R}}D_i$. As $\frac{\partial}{\partial x}\area(P\cap(Q+(x,y^*))) = B+Dy^*$ is constant for all $x\in (x_l,x_r)$, at least one of $\area(P\cap(Q+(x_l,y^*)))$ and $\area(P\cap(Q+(x_r,y^*)))$ is at least $\area(P\cap(Q+(x^*,y^*)))$. As both $x_l$ and $x_r$ are in $X$, this implies the claim.
\end{proof}

Next, we represent the function $\tau\mapsto\area(P\cap(Q+\tau))$ via $\O(nm)$ \textit{translation slabs}, i.e., unbounded rectangular areas $[l,r)\times[b,\infty)$ endowed with a bilinear function, such that the function $\tau\mapsto\area(P\cap(Q+\tau))$ is the sum of bilinear functions of translation slabs containing $\tau$ evaluated at $\tau$.

\begin{definition}[translation slabs]
    Let $P$ and $Q$ be two orthogonal polygons. A set of half-open rectangles $\{[l_i,r_i)\times[b_i,\infty)\mid i\in I\}$ each endowed weights $A_i$, $B_i$, $C_i$, and $D_i$ for $i\in I$ is called a set of translation slabs of $P$ and $Q$ if for every $(x,y) \in \bR^2$
    \begin{equation*}
        \area(P\cap(Q+(x,y)))=\sum_{i\in I}\mathbbm{1}[(x,y)\in [l_i,r_i)\times[b_i,\infty)]\cdot\left(A_i + B_i x + C_i y + D_i xy\right).
    \end{equation*}
\end{definition}

\begin{lemma}\label{lem:computeSlabs}
    Let $P$ and $Q$ be two orthogonal polygons with $n$ and $m$ vertices. Let $P_x,P_y,Q_x$ and $Q_y$ be the set of $x$- and $y$-coordinates of $P$ and $Q$ respectively. Let $X=P_x-Q_x$, $Y=P_y-Q_y$. There is an algorithm which computes a set of $\O(nm)$ translation slabs of $P$ and $Q$ in $\O(nm+(n+m)\log(n+m))$ time, where every translation slab is defined by $l,r\in X$ and $b\in Y$ as $[l,r)\times[b,\infty)$.
\end{lemma}
\begin{proof}
    Via \Cref{lem:polygonCutting}, the algorithm first computes a decomposition of $P$ and $Q$ into $\O(n)$ and $\O(m)$ rectangles in $\O((n+m)\log (n + m))$ total time. By slight abuse of notation let us assume that they are exactly $n$ and $m$ rectangles. Let $P_1,\ldots,P_n$ be the rectangles from the decomposition of $P$ and $Q_1,\ldots,Q_m$ be the rectangles from the decomposition of $Q$. Then for every $(x,y)\in\bR^2$
    \[\area(P\cap(Q+(x,y))=\sum_{i=1}^n\sum_{j=1}^m\area(P_i\cap(Q_j+(x,y)).\]
    For every $i\leq n$ and $j\leq m$, let $R_{i,j,1},\ldots,R_{i,j,K}$ be the $K\leq 18$ rectangular half-open regions from \Cref{cor:translationslabs} with values $A_{i,j,k},B_{i,j,k},C_{i,j,k}$ and $D_{i,j,k}$ for $1\leq k\leq K$. Then 
    \[\area(P\cap(Q+(x,y))=\sum_{i=1}^n\sum_{j=1}^m\sum_{k=1}^K\mathbbm{1}[(x,y)\in R_{i,j,k}]\left(A_{i,j,k} + B_{i,j,k}x+C_{i,j,k}y+D_{i,j,k}xy\right).\]
    Note that by \Cref{lem:polygonCutting} these translation slabs can be computed in time $\O(nm+(n+m)\log(n+m))$ total time.
\end{proof}

Finally, with the representation of $\tau\mapsto\area(P\cap(Q+\tau))$ from \Cref{lem:computeSlabs}, we present an algorithm solving \rectPolyProb in $\tildeO((nm)^2)$ time. 
For this, denote the set of translation slabs output by the algorithm from \Cref{lem:computeSlabs} by $\mathcal{R}(P,Q)$, where $R\in\mathcal{R}(P,Q)$ is endowed with the weights $A_R,B_R,C_R$ and $D_R$. 
Note that $\mathcal{R}(P,Q)$ is a set of rectangular regions in \emph{translation space}, unlike the set of rectangles obtained by the decomposition of both $P$ and $Q$ via \Cref{lem:polygonCutting}.

\begin{lemma}\label{lem:query}
    Let $P$ and $Q$ be orthogonal polygons with $n$ and $m$ vertices respectively. Let $P_x,P_y,Q_x$ and $Q_y$ be the set of $x$- and $y$-coordinates of $P$ and $Q$ respectively. Let $X=P_x-Q_x$, $Y=P_y-Q_y$ and let $\mathcal{Q}\subset X\times Y$ be given. There is an algorithm that for every $q\in \mathcal{Q}$ computes \(A_q=\sum_{R\in\mathcal{R}(P,Q)}\mathbbm{1}[q\in R]A_R\), \(B_q=\sum_{R\in\mathcal{R}(P,Q)}\mathbbm{1}[q\in R]B_R\),\(C_q=\sum_{R\in\mathcal{R}(P,Q)}\mathbbm{1}[q\in R]C_R\), and \(D_q=\sum_{R\in\mathcal{R}(P,Q)}\mathbbm{1}[q\in R]D_R\),
    in total time $\O(nm\log(nm) + |Q|\log(nm))$.
\end{lemma}
\begin{proof}
    We imagine sweeping a horizontal line from $-\infty$ to $\infty$ over the values of $Y$. During this sweep we maintain a binary tree with leaves corresponding to the values in $X$. 
    For any node $v$ of the tree let $X_v$ be the set of leaves in the tree rooted at $v$. At a given sweep line event $y$ for a given node $v$ with parent $p$ let $\mathcal{R}_v$ be the subset of the translations slabs $R\in \mathcal{R}(P,Q)$ such that $X_v\times\{y\}\subset R$ and $X_p\times\{y\}\not\subset R$.
    For any node $v$ we maintain the values $A_v=\sum_{r\in \mathcal{R}_v}A_r$, $B_v=\sum_{r\in \mathcal{R}_v}B_r$, $C_v=\sum_{r\in \mathcal{R}_v}C_r$ and $D_v=\sum_{r\in \mathcal{R}_v}D_r$. Observe that initially all $A_v=B_v=C_v=D_v=0$ can correctly be computed. Now, at every sweep-line event $y$ for every $R=[l,r)\times[b,\infty)\in\mathcal{R}(P,Q)$ such that $b=y$ we add $A_R$, $B_R$, $C_R$ and $D_R$ to $A_v$, $B_v$, $C_v$ and $D_v$ of the roots $v$ of the $\O(\log(nm))$ maximal subtrees contained in $[l,r)$ maintaining the values correctly, as by \Cref{lem:polygonCutting} every translation slab has its lower boundary at some $y\in Y$. This maintenance takes $\O(nm\log(nm))$ total time.
    
    For $y\in Y$ let $\mathcal{Q}_y$ be the subset of elements $q=(q_x,q_y)$ in $\mathcal{Q}$ such that $q_y=y$. When the sweep-line has finished handling all translation slabs whose lower boundary is $y$, we can output $A_q$, $B_q$, $C_q$ and $D_q$ for any $q=(q_x,q_y)\in \mathcal{Q}_y$ by traversing the tree from root to the leaf representing $q_x$ summing all $A_v$, $B_v$, $C_v$ and $D_v$ along the path. As this path has length $\O(\log(nm))$, the values $A_q$, $B_q$, $C_q$ and $D_q$ can be computed in time $\O(\log(nm))$ per element in $\mathcal{Q}_y$ and hence $\mathcal{Q}$ throughout the entire sweep-line algorithm concluding the proof. 
\end{proof}

\begin{corollary}\label{cor:trivial}
    Let $P$ and $Q$ be orthogonal polygons with $n$ and $m$ vertices respectively. There is an algorithm that computes $\max_{\tau\in\bR^2}\area(P\cap(Q+\tau))$ in $\O((nm)^2\log(nm))$ time.
\end{corollary}
\begin{proof}
    Let $P_x,P_y,Q_x$ and $Q_y$ be the set of $x$- and $y$-coordinates of $P$ and $Q$ respectively. Let $X=P_x-Q_x$, $Y=P_y-Q_y$.
    By \Cref{lem:discretization}, the maximum is attained by a translation $\tau\in X\times Y$.
    The claim is then a consequence of \Cref{lem:query} by setting $\mathcal{Q}\gets X\times Y$.
    This gives us for every $q=(x,y)\in X\times Y$ values $A_q,B_q,C_q$ and $D_q$ such that \(\area(P\cap(Q+q))=A_q+B_qx+C_qy+D_qxy\) can be evaluated in $\O(1)$ time, concluding the proof.
\end{proof}

\section{\boldmath An $\O((nm)^{3/2} \log (nm))$ Algorithm}\label{sec:fast}

\begin{algorithm}[!b]
\caption{Maximum Polygon Overlap}
\label{alg}
\begin{algorithmic}[1]
\Procedure{MaxRectPolygonOverlap}{$P$, $Q$}
    \State Let $P_x,P_y,Q_x$ and $Q_y$ be the set of $x$- and $y$-coordinates of $P$ and $Q$ respectively\label{line:test}
    \State Let $X\gets P_x-Q_x$, $Y\gets P_y-Q_y$, and store them in sorted lists
    \State Compute the set of translation slabs $\mathcal{R}(P,Q)$ together with its weights
    \State Associate every translation slab $[l,r)\times[b,\infty)\in\mathcal{R}(P,Q)$ to $b\in Y$
    \State Partition $Y$ into $C=\O(\sqrt{|P||Q|})$ contiguous sets $Y_1,\ldots,Y_C$ such that each $Y_\beta$ is of minimal size, with at least $\sqrt{|P||Q|}$ translation slabs associated to its elements\label{line:grouping}
    \State For every $Y_\beta$ let $y_\beta^+$ be its last element
    \State Let $Y_\text{heavy}=\{y_\beta^+\mid \beta\leq C\}$, and $Y_\beta\gets Y_\beta\setminus\{y_\beta^+\}$ for every $\beta\leq C$\label{line:cuttingLast}
    \State Let $\tau$ be the element in $X\times Y_\text{heavy}$ where $\area(P\cap(Q+\tau))$ is maximum, via \Cref{lem:query}\label{line:computeHeavy1}
    \State $A\gets \area(P\cap(Q+\tau))$\label{line:computeHeavy2}
    \For{$\beta=1,\ldots,C$}\label{line:betaloop}\Comment{\emph{outer loop}}
        \State Let $y_1\in Y_\beta$ be the first element in $Y_\beta$ and $\mathcal{Q}= X\times\{y_1\}$\label{line:prereqStart}
        \State Compute $A_q$, $B_q$, $C_q$ and $D_q$ for every $q\in \mathcal{Q}$, via \Cref{lem:query}\label{line:Aq}
        \State Let $\mathcal{R}_\beta$ be the translation slabs associated to elements in $Y_\beta$
        \State Compute the arrangement $\mathcal{A}_\beta$ of $\{[l,r)\mid [l,r)\times[b,\infty)\in\mathcal{R}_\beta\}$\label{line:prereqEnd}
        \State Partition $X$ into contiguous sets $X_1,\ldots,X_{C'}$ according to $\mathcal{A}_\beta$\label{line:xsplit} 
        \For{$\alpha=1,\ldots,C'$}\label{line:alphaloop}\Comment{\emph{inner loop}}
            \State Let $\mathcal{R}_{\alpha,\beta}$ be the set of translation slabs $[l,r)\times[b,\infty)$ in $\mathcal{R}_\beta$ with $X_\alpha\subset [l,r)$
            \For{$y_j\in Y_\beta$}\Comment{\emph{transform $Y_\beta$ to $V\subset\bR^3$}}
                \State Let $\mathcal{R}_{\alpha,\beta,j}$ be the slabs $[l,r)\times[b,\infty)$ in $\mathcal{R}_{\alpha,\beta}$ with $y_1<b\leq y_j$\label{line:Rabj}
                \State Let $v_j=((y_j-y_1),\sum_{R\in \mathcal{R}_{\alpha,\beta,j}}(B_R + D_Ry_j),\sum_{R\in \mathcal{R}_{\alpha,\beta,j}}(A_R + C_Ry_j)^T)$
            \EndFor
            \State Let $V=\{v_j\mid y_j\in Y_\beta\}$ 
            \For{$x_i\in X_\alpha$}\label{line:innerloop}\Comment{\emph{Query $V$ instead of $Y_\beta$}}
                \State Compute $M=\max_{v_j\in V}\langle v_j,(C_q+D_qx_i,x_i,1)^T\rangle$ realized by $v_{j^*}\in V$
                \State $A' = A_q + B_qx_i + C_qy_1 + D_qx_iy_1 + M$\label{line:aprime2}
                \If{$ A' > A$} $A\gets A'$, $\tau\gets (x_i,y_{j^*})$\EndIf
            \EndFor
        \EndFor
    \EndFor
    \State \textbf{Return} $\tau$, $A$
\EndProcedure
\end{algorithmic}
\end{algorithm}


In this section we give critical insights that allow batching queries, improving on the algorithm given in \Cref{cor:trivial}. The improved algorithm splits the set $X\times Y$ from \Cref{lem:discretization} into $\O(\sqrt{nm})$ horizontal strips $X\times Y_\beta$ each $Y_\beta$ containing $\O(\sqrt{nm})$ elements of $Y$. In each such strip there are $\O(\sqrt{nm})$ lower boundaries of translation slabs $[l_i,r_i)\times[b_i,\infty)\in \mathcal{R}(P,Q)$. 
The corresponding intervals $[l_i,r_i)$ split $X\times Y_\beta$ into vertical strips $X_\alpha\times Y_\beta$, each containing $\O(\sqrt{nm})$ elements of $X$ on average.
For each such $X_\alpha\times Y_\beta$, we encode $Y_\beta$ in a three dimensional convex hull.
For any $x\in X_\alpha$, we need to find the $y\in Y_\beta$ realizing $\max_{y\in Y_\beta}\area(P\cap(Q+(x,y))$. This then amounts to computing an extreme point of the convex hull `in the direction of $x$'. Importantly, in three dimensions one can compute and store the convex hull in near-linear time allowing extreme-point queries in logarithmic time.
\begin{lemma}[\cite{dobkin1985linear,dobkin1990separation}]\label{lem:convexhull}
    Given $n$ points $p_1,\ldots,p_n$ in $\bR^3$ one can compute in $\O(n\log n)$ time a data structure that can answer extreme point queries for points $q\in \bR^3$ of the form
    $\mathrm{argmax}_i\langle p_i,q\rangle$
    in time $\O(\log n)$.
\end{lemma}


\begin{lemma}\label{lem:correctness}
    Let $P$ and $Q$ be orthogonal polygons with $n$ and $m$ vertices respectively. \Cref{alg} provided with $P$ and $Q$ computes some $\tau^*\in \bR^2$ and $A^*\in \bR$ such that $\tau^*$ attains the maximum overlap $\max_{\tau\in\bR^2}\area(P\cap(Q+\tau)) = A^* = \area(P\cap(Q+\tau^*))$.
\end{lemma}
\begin{proof}
     By \Cref{lem:discretization} it suffices for the algorithm to compute $\max_{\tau'\in X\times Y}\area(P\cap(Q+\tau'))$.
     
    First observe that $Y$ in line \ref{line:grouping} can indeed be partitioned as claimed via a linear scan of $Y$ as there are a total of $\O(nm)$ translation slabs in $\mathcal{R}(P,Q)$. 
    Hence after line \ref{line:cuttingLast} $Y_\text{heavy}$ and all $Y_\beta$ together partition $Y$, where each $Y_\beta$ is a contiguous subset of $Y$. \Cref{lem:query} implies that after line \ref{line:computeHeavy1} and line \ref{line:computeHeavy2} $A$ stores $\area(P\cap(Q+\tau))=\max_{\tau'\in X\times Y_\text{heavy}}\area(P\cap(Q+\tau'))$.
    
    
    As the sets $Y_\beta$ partition $Y\setminus Y_\text{heavy}$, and for any fixed $Y_\beta$ the sets $X_\alpha$ partition $X$, it suffices to show that for any $\alpha$, $\beta$, and $x_i\in X_\alpha$ the value $A'$ computed in line \ref{line:aprime2} corresponds to $\max_{y_j\in Y_\beta}\area(P\cap (Q+(x_i,y_j)))$. 
    For this, consider some iteration of the outer \textbf{for} loop in line \ref{line:betaloop} and fix some $\alpha\leq C'$. Next, observe that $\mathcal{R}_{\alpha,\beta}$ is exactly the set of translation slabs $[l,r)\times[b,\infty)$ in $\mathcal{R}(P,Q)$ which has its lower boundary $b$ coincide with some $y_j\in Y_\beta$ and $X_\alpha\subset [l,r)$. In particular, this implies that for the first element $y_1$ in $Y_\beta$ and any $x_i\in X_\alpha$ and $y_j\in Y_\beta$ the set $\{R\in \mathcal{R}(P,Q)\mid (x_i,y_j)\in R\}$ decomposes into $\{R\in \mathcal{R}(P,Q)\mid (x_i,y_1)\in R\}$ and $\mathcal{R}_{\alpha,\beta,j}=\{[l,r)\times[b,\infty)\in\mathcal{R}_{\alpha,\beta}\mid y_1 < b\leq y_j\}$. And thus for $q=(x_i,y_1)\in\mathcal{Q}$ we have  
    {\allowdisplaybreaks
    \begin{align*}
        &\area(P\cap(Q+(x_i,y_j))) =\sum_{R\in\mathcal{R}(P,Q)}\mathbbm{1}[(x_i,y_j)\in R]\left(A_{R}+B_{R}x_i+C_{R}y_j+D_{R}x_iy_j\right)\\
        =&\sum_{R\in\mathcal{R}(P,Q)}\mathbbm{1}[q\in R]\left(A_{R}+B_{R}x_i+C_{R}y_j+D_{R}x_iy_j\right)+\sum_{R\in\mathcal{R}_{\alpha,\beta,j}}\left(A_{R}+B_{R}x_i+C_{R}y_j+D_{R}x_iy_j\right)\\
        =&A_q+B_qx_i+C_qy_j+D_qx_iy_j+\sum_{R\in\mathcal{R}_{\alpha,\beta,j}}\left(A_{R}+B_{R}x_i+C_{R}y_j+D_{R}x_iy_j\right)
    \end{align*}
    }
    
    
    for the values $A_q$, $B_q$, $C_q$ and $D_q$ computed for $q$ in line \ref{line:Aq} via \Cref{lem:query}. Finally as
    {
    \allowdisplaybreaks
    \begin{align*}
        &\max_{y_j\in Y_\beta} \left(\area(P\cap(Q+(x_i,y_j)))\right)\\
        =&\max_{y_j\in Y_\beta} \Bigg(A_q + B_qx_i + C_qy_j + D_qx_iy_j + \sum_{R\in \mathcal{R}_{\alpha,\beta,j}}\left(A_R+B_Rx_i+C_Ry_j + D_Rx_iy_j\right)\Bigg)\\
        =&\max_{y_j\in Y_\beta}\Bigg(A_q + B_qx + C_qy_1 + C_q(y_j-y_1) + D_qxy_1 + D_qx(y_j-y_1)\\
        &+ \sum_{R\in \mathcal{R}_{\alpha,\beta,j}}\left(A_R+B_Rx_i+C_Ry_j + D_Rx_iy_j\right)\Bigg)\\
        =&A_q + B_qx_i + C_qy_1 + D_qx_iy_1 \\
        &+ \max_{y_j\in Y_\beta}\Bigg(\!(C_q+D_qx_i)(y_j-y_1) + \bigg(\sum_{R\in \mathcal{R}_{\alpha,\beta,j}}\left(B_R + D_Ry_j\right)\!\!\bigg)x_i + \sum_{R\in \mathcal{R}_{\alpha,\beta,j}}\left(A_R+C_Ry_j\right)\!\!\Bigg)\\
        =&A_q + B_qx_i + C_qy_1 + D_qx_iy_1 + \max_{y_j\in Y_{\beta}}\left\langle\begin{pmatrix}y_j-y_1\\
        \sum_{R\in \mathcal{R}_\alpha,i\leq j}(B_R + D_Ry_j)\\
        \sum_{R\in \mathcal{R}_\alpha,i\leq j}(A_R + C_Ry_j)  
        \end{pmatrix},\begin{pmatrix}
            C_q + D_qx_i\\
            x_i\\
            1
        \end{pmatrix}\right\rangle,
    \end{align*}
    }
    the computed value $A'$ in line \ref{line:Aq} is exactly $\max_{y_j\in Y_\beta} \left(\area(P\cap(Q+(x_i,y_j)))\right)$.
\end{proof}

We remark that the sole purpose of $Y_\text{heavy}$ is to guarantee that each $Y_\beta$ has at most $\sqrt{nm}$ translation slabs associated to it, and in turn the arrangement $\mathcal{A}_\beta$ consists of at most $\O(\sqrt{nm})$ contiguous subsets of $X$. This is important, as we transform $Y_\beta$ into a different set $V$ for each contiguous subset in $\mathcal{A}_\beta$. We note that by construction any $y\in Y$ that has at least $\sqrt{nm}$ translation slabs associated to it will end up in $Y_\text{heavy}$. As there are at most $\O(nm)$ translation slabs, $Y_\text{heavy}$ cannot be too big, and hence $\area(P\cap(Q+\tau))=\max_{\tau'\in X\times Y_\text{heavy}}\area(P\cap(Q+\tau'))$ can be computed efficiently via \Cref{lem:query}. In fact, as every $y\in Y$ has at least one translation slab associated to it (\Cref{lem:translationfunction} and \Cref{cor:translationslabs}), $Y_\text{heavy}$ can be seen as just being a safe guard against inputs that are not in `general position', where $|Y|\ll nm$. The exponent of $3/2$ is the result of balancing the sizes of $Y_\beta$ with the sizes of $X_\alpha$.

\begin{lemma}\label{lem:runningtime}
    Let $P$ and $Q$ be orthogonal polygons with $n$ and $m$ vertices respectively. The algorithm described in \Cref{alg} provided with $P$ and $Q$ has running time $\O((nm)^{3/2}\log(nm))$.
\end{lemma}
\begin{proof}
    First observe that $Y$ in line \ref{line:grouping} can be partitioned via a linear scan in $\O(nm)$ time. 
    Furthermore, observe that for any $\beta\leq C$ the set $Y_\beta\setminus\{y_\beta^+\}$ has less than $\sqrt{nm}$ translation slabs associated to it, as otherwise $Y_\beta$ would not have been minimal.
    Hence after line \ref{line:cuttingLast} the set $Y_\text{heavy}$ and all sets $Y_\beta$ together partition $Y$, $|Y_\text{heavy}|\leq\O(\sqrt{nm})$, and each $Y_\beta$ has at most $\sqrt{nm}$ translation slabs associated to it. Note that $|Y_\beta|\leq\O(\sqrt{nm})$ as any $y\in Y$ has at least one translation slab associated to it (\Cref{lem:translationfunction} and \Cref{cor:translationslabs}).
    By \Cref{lem:query}, line \ref{line:computeHeavy1} and line \ref{line:computeHeavy2} take a total of $\O((nm)^{3/2}\log(nm))$ time.

    Now consider one iteration of the outer \textbf{for} loop in line \ref{line:betaloop}. Lines \ref{line:prereqStart}--\ref{line:prereqEnd} take a total of $\O(\sqrt{nm}\log(nm))$ time, by \Cref{lem:query} together with the fact that $|\mathcal{R}_\beta|\leq\sqrt{nm}$. With $\mathcal{A}_\beta$ at hand, computing the partition of $X$ in line \ref{line:xsplit} takes $\O(nm)$ time.
    
    Now consider one iteration of the inner \textbf{for} loop in line \ref{line:alphaloop}. $\mathcal{R}_{\alpha,\beta}$ can be computed in $\O(|\mathcal{R}_\beta|)=\O(\sqrt{nm})$ time. Similarly, for every $y_j\in Y_\beta$ the values $\sum_{R\in \mathcal{R}_{\alpha,\beta,j}}A_R$, $\sum_{R\in \mathcal{R}_{\alpha,\beta,j}}B_R$, $\sum_{R\in \mathcal{R}_{\alpha,\beta,j}}C_R$, and $\sum_{R\in \mathcal{R}_{\alpha,\beta,j}}D_R$ and with it the set $V$ can be computed in $\O(\sqrt{nm})$ time via a linear scan of $Y_\beta$, as $\mathcal{R}_{\alpha,\beta,j}$ and $\mathcal{R}_{\alpha,\beta,j+1}$ differ by exactly all translation slabs in $\mathcal{R}_{\alpha,\beta}$ associated to $y_{j+1}$. 
    By \Cref{lem:convexhull}, endowing the set $V$ with a data structure allowing three-dimensional extreme-point queries can be done in time $\O(\sqrt{nm}\log(nm))$. Finally, each iteration of the \textbf{for} loop in line \ref{line:innerloop} takes $\O(\log(nm))$ time, computing $\max_{v_j\in V}\langle v_j,(C_q+D_qx_i,x_i,1)^T\rangle$ via the computed data structure in $\O(\log(|V|))$ time by \Cref{lem:convexhull}. Thus each iteration of the inner \textbf{for} loop in line \ref{line:alphaloop} takes $\O(\sqrt{nm}\log (nm) + |X_\alpha|\log(nm))$ time. This implies that the iteration of the outer \textbf{for} loop in line \ref{line:betaloop} takes total time in
    \[\O\left(nm + \sqrt{nm}\log(nm) + \sum_{\alpha\leq C'}\left(\sqrt{nm}\log(nm) + |X_\alpha|\log(nm)\right)\right)=\O(nm\log(nm)).\]
    Finally the total running time of the algorithm is bound by
    \[\O\left((nm)^{3/2}\log(nm) + \sum_{\beta\leq C}\left(nm\log(nm)\right)\right)=\O\left((nm)^{3/2}\log(nm)\right).\]
    
\end{proof}

\fastAlg*
\begin{proof}
    This is an immediate consequence of \Cref{lem:correctness} and \Cref{lem:runningtime}.
\end{proof}

\section{Lower Bounds}
We now complement our algorithmic results by giving a lower bound.
We reduce to the Polygon Overlap problem from a slightly more general form of $5$-SUM, which we call the imbalanced $5$-SUM, or $(3,2)$-SUM:

\begin{problem}[$(3,2)$-SUM]
    Given three sets of positive integers $A$, $B$, and $C$, of size $n$ each, and two sets $D$, and $E$ of size $m\leq n$ each, is there $(a,b,c,d,e)\in A\times B\times C\times D \times E$ such that $a=b+c+d+e$?
\end{problem}

\begin{lemma}\label{lem:32sum}
    Assuming the $k$-SUM Hypothesis, then $(3,2)$-SUM with sets of size $n$ and $m\leq n$ cannot be solved in $\O((\max(n^2,nm^2))^{1-\eps})$ time for any $\eps>0$.
\end{lemma}
\begin{proof}
    First, observe that any $3$-SUM instance can be transformed into a $(3,2)$-SUM instance with $m=1$, hence $(3,2)$-SUM cannot be solved in $\O(n^{2-\eps})$ time.

    Next, observe that any $5$-SUM instance can be split into $\Theta((n/m)^2)$ $(3,2)$-SUM instances where the original $5$-SUM instance has a solution if and only if one of the $(3,2)$-SUM instances has a solution. This is the case by splitting two sets of the $5$-SUM instance into $\lceil\frac{n}{m}\rceil$ sets of size $m$ each, forming one $(3,2)$-SUM instance for every combination of these smaller sets. Hence, by the pidgeon hole principle, at least one $(3,2)$-SUM instance cannot be solved in $\O(n^{3-\eps}/(n/m)^2)=\O((nm^2)^{1-\eps})$ time concluding the proof.
\end{proof}

\subsection{Construction}

Let $A$, $B$, $C$, $D$ and $E$ constitute a $(3,2)$-SUM instance. We assume the sets to be sorted. We now construct two polygons $P$ and $Q$ with $\O(n)$ and $\O(m)$ vertices respectively, together with a value $\alpha$ such that deciding if $\max_{\tau\in \bR^2}\area(P\cap(Q+\tau))\geq \alpha$ amounts to verifying that the $(3,2)$-SUM instance has a solution.

For the construction we fix two parameters. The first parameter $M$ we choose to be $100\sum_{S\in \{A,B,C,D,E\}}\sum_{s\in S} S$. The parameter $M$ will serve as a trivial upper bound to any solution. The second parameter $\eps$ we choose to be $1/(100n^2)$, serving as a sufficiently small grid size on which we construct the polygons\footnote{Note that our choices of both $M$ and $\eps$ are quite conservative. We believe the construction works with $M=\Theta(\max_{S\in \{A,B,C,D,E\}}\max_{s\in S} S)$ and $\eps=\Theta(1/n)$ as well.}.

Observe that we can connect any set of $n$ pairwise disjoint polygons to form a single polygon via $\O(n)$ polygons with measure $0$. As such, we may describe the polygons $P$ and $Q$ as a union of disjoint polygons instead\footnote{We remark that we use measure $0$ connections for ease of exposition. The same construction also works with connecting pieces with non-zero measure, as long as the total measure $\mu$ of all connecting pieces is sufficiently small, i.e., in particular $\mu\ll\eps^2$.}. Both $P$ and $Q$ consist of four gadgets each. We give the precise construction of the gadgets first, before giving intuition on the constructed instance. We conclude proving that computing the maximum overlap between these polygons amounts to solving $(3,2)$-SUM. For the construction of the gadgets (except for the anchor gadget) refer to \Cref{fig:reduction}.

\textbf{Anchor gadget:} Both $P$ and $Q$ have an anchor gadget, whose sole purpose is to restrict the set of translations which may attain the maximum. The anchor gadget $P_A$ of $P$ is a square with side length $M$ including the anchor gadget $Q_A$ of $Q$ being a square of side length $1$ sufficiently far away from the origin. More precisely, the anchor gadget of $P$ is defined by the vertices $(0,100M)$, $(0,101M+1)$, $(M,101M+1)$ and $(M,100M)$, and the anchor gadget of $Q$ is defined by the vertices $(0,100M)$, $(0,100M+1)$, $(1,100M+1)$ and $(1,100M)$. Observe in particular that the anchor gadgets lie on the grid with side length $\eps$.

\begin{figure}
    \centering
    \includegraphics[width=0.8\textwidth]{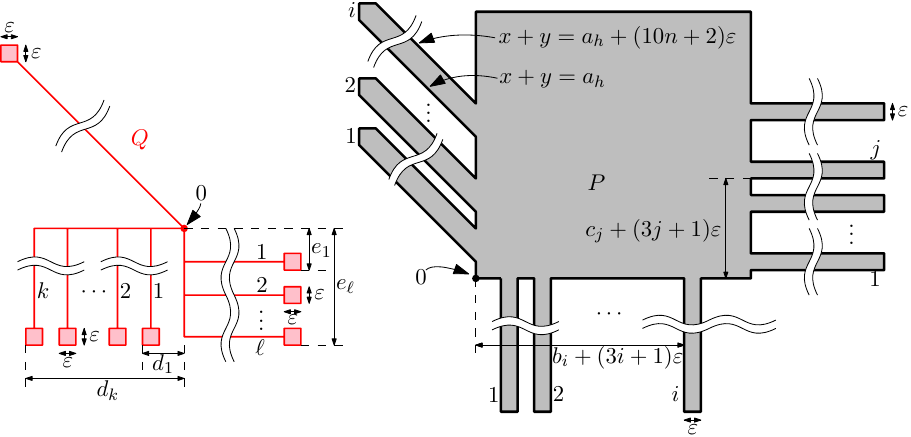}
    \caption{Construction of the polygons $P$ and $Q$ such that maximizing $\area(P\cap(Q+\tau))$ amounts to solving a $(3,2)$-SUM instance. The anchor gadgets of $P$ and $Q$ are not shown.}
    \label{fig:reduction}
\end{figure}

\textbf{Origin:} Unlike the anchor gadget, which is far away from the origin, all other pieces of non-zero measure of $P$ and $Q$ are contained in $[-2M,3M]\times[-2M,3M]$, as we will see later. The rectangle $[0,M]\times[0,M]$ is part of $P$ and serves as a root polygon to which its other gadgets are attached.

\textbf{$x$-comb gadget:} For the $\bindex$\ts{th} element $b_{\bindex}\in B$ add the rectangle $P_x^{\bindex}=[b_{\bindex} + (3{\bindex}+1)\eps,b_{\bindex}+(3{\bindex}+2)\eps]\times[-2M,0]$ to $P$, which we refer to as the ${\bindex}$\ts{th} $x$-prong of $P$. 
For the ${\dindex}$\ts{th} element $d_{\dindex}\in D$ we add the square $Q_x^{\dindex}[-d_{\dindex},-d_{\dindex}+\eps]\times[-2M,-2M+\eps]$ to $Q$, which we refer to as the ${\dindex}$\ts{th} $x$-prong of $Q$.

\textbf{$y$-comb gadget:} For the ${\cindex}$\ts{th} element $c_{\cindex}\in C$ add the rectangle $P_y^{\cindex} = [M,3M]\times [c_{\cindex} + (3{\cindex}+1)\eps,c_{\cindex}+(3{\cindex}+2)\eps]$ to $P$, which we refer to as the ${\cindex}$\ts{th} $y$-prong of $P$. 
For the ${\eindex}$\ts{th} element $e_{\eindex}\in E$ we add the square $Q_y^{\eindex}=[2M-\eps,2M]\times [-e_{\eindex},-e_{\eindex}+\eps]$ to $Q$, to which we refer as the ${\eindex}$\ts{th} $y$-prong of $Q$.

\textbf{sum gadget:} For the ${\aindex}$\ts{th} element $a_{\aindex}\in A$ let $P_d^{\aindex}$ be the polygon resulting from restricting the set $\{(x,y)\in\bR^2\mid a_{\aindex}\leq x+y\leq a_{\aindex}+(10n+2)\eps\}$ to $([-2M,M]\times[0,3M])\setminus([0,M]\times[0,M])$, and add it to $P$. We call the polygon corresponding to $a_{\aindex}$ the ${\aindex}$\ts{th} diagonal prong of $P$. Lastly add the single square $Q_v=[-2M,-2M+\eps]\times[2M,2M+\eps]$ to $Q$, called the verifier of $Q$.

\textbf{Intuition:}
The role of the anchor gadgets is to restrict the set of translations $\tau$ such that $\area(P\cap(Q+\tau))>1$ to $[0,M]\times[0,M]$. This is achieved since outside the anchor gadget of $Q$ (which has measure exactly $1$), $Q$ has a total measure of $(2n+1)\eps\ll 1$. Hence, any translation that attains an overlap of more than $1$ must place the anchor gadget of $Q$ within the anchor gadget of $P$ and move at least one prong of $Q$ inside $P$.

The role of the $x$-prongs is slightly more involved: By construction of the anchor gadgets, any $x$-prong of $Q$ may only be translated to intersect an $x$-prong of $P$. By the $(3{\bindex})\eps$-shifts applied to the $x$-prongs of $P$, at most one $x$-prong of $Q$ can ever intersect any $x$-prong of $P$. Thus, the area attained from intersections between $x$-prongs of $Q$ and $P$ is in $[0,\eps^2]$, and in particular non-zero if and only if the $x$-coordinate of the translation is (approximately) $b_{\bindex}+d_k$ for some $j$ and $k$. The role of the $y$-prongs is analogous and guarantees that at the area of the intersection between the $y$-prongs of $Q$ and $P$ is in $[0,\eps^2]$ and only non-zero if the $y$-coordinate of the translation is (approximately) $c_{\cindex}+e_{\eindex}$. 

Similarly to the $x$-prongs and the $y$-prongs, the verifier of $Q$ can intersect $P$ only in its diagonal prongs. In particular, if the verifier intersects the $\aindex$\ts{th} diagonal prong, then the sum of the $x$- and $y$-coordinates of the translation moving the verifier into the prong is (approximately) equal to $a_{\aindex}$. Hence in total, we can (approximately) verify that there is a $(a_{\aindex},b_{\bindex},c_{\cindex},d_{\dindex},e_{\eindex})\in A\times B\times C\times D\times E$ with $a_{\aindex}=b_{\bindex}+c_{\cindex}+d_{\dindex}+e_{\eindex}$ by checking if the maximum overlap is at least (or rather exactly) $1+3\eps^2$ and strictly less than $1+3\eps^2$ otherwise.

\subsection{Correctness}

\begin{figure}[!t]
    \centering
    \includegraphics[width=0.8\textwidth]{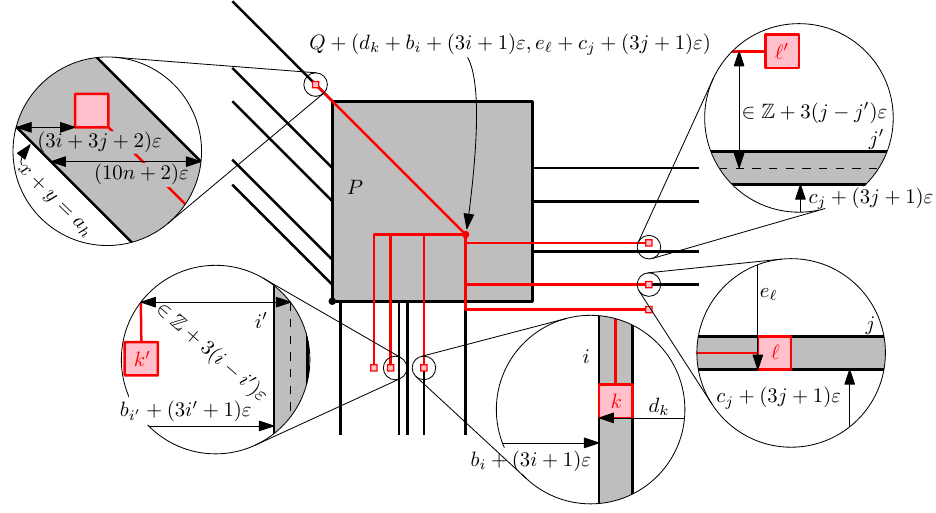}
    \caption{Illustration of the correctness of the $5$-SUM lower bound with the $\eps$-shifts for a given translation $(d_{\dindex}+b_{\bindex}+(3{\bindex}+1)\eps,e_{\eindex}+c_{\cindex}+(3{\cindex}+1)\eps)$. The highlighted areas starting from right to left in a clock-wise order correspond (i) for no ${\eindex}'\neq {\eindex}$ can the ${\eindex}'$\ts{th} $y$-prong of $Q$ align with a $y$-prong of $P$ if $\eps\ll 1/n$, (ii) the ${\eindex}$\ts{th} $y$-prong of $Q$ aligns with the ${\cindex}$\ts{th} $y$-prong of $P$,(iii) the ${\dindex}$\ts{th} $x$-prong of $Q$ aligns with the ${\bindex}$\ts{th} $x$-prong of $P$, (iv) for no ${\dindex}'\neq {\dindex}$ can the ${\dindex}'$\ts{th} $x$-prong of $Q$ align with a $y$-prong of $P$ if $\eps\ll 1/n$, and (v) if there is an $a_{\aindex}$ with $a_{\aindex}=b_{\bindex}+c_{\cindex}+d_{\dindex}+e_{\eindex}$, then the verifier of $Q$ lies in the $\aindex$\ts{th} diagonal prong of $P$. Otherwise, if $\eps\ll 1/n$, then the diagonal prong of $Q$ lies outside of $P$.
    }
    \label{fig:reductionCorrectness}
\end{figure}

We proceed as described above, first restricting the translations attaining an area of overlap of at least $1$ to translations in $[0,M]\times[0,M]$, and then further restricting the translations via the careful positioning of the prongs.

\begin{lemma}\label{lem:anchor}
    Let $P$ and $Q$ be polygons as constructed previously. Any translation $\tau$ such that $\area(P\cap(Q+\tau))>1$ must be in $[-1,M+1]\times[-1,M+1]$.    
\end{lemma}
\begin{proof}
    By construction of $Q$, $\area(P\cap(Q+\tau))>1$ only if $\area(P\cap(Q_A+\tau))>0$ and $\area(P\cap((Q\setminus Q_A)+\tau))>0$. The first is only possible if $Q_A$ intersects either $P_A$ or $[-2M,3M]\times[-2M,3M]$, i.e., either $\tau\in ([-1,M+1]\times[-1,M+1])$ or $([-102M-1,-97M]\times[-2M-1,3M]))$. However, in the second case $(Q\setminus Q_A)+\tau$ lies below the line $y=-50M$. As $P$ lies completely above the line $y=-2M$, $\tau\in ([-1,M+1]\times[-1,M+1])$.
\end{proof}

We now concern ourselves with all other gadgets and their interactions. For the rest of this section, refer to \Cref{fig:reductionCorrectness}.

\begin{lemma}\label{lem:prongprong}
    Let $\tau\in [-1,M+1]\times[-1,M+1]$. Then
    \begin{enumerate}[(i)]
        \item $Q_v+\tau$ can only intersect diagonal prongs of $P$,
        \item Any $Q_x^{\dindex}+\tau$ can only intersect $x$-prongs of $P$, and
        \item Any $Q_y^{\eindex}+\tau$ can only intersect $y$-prongs of $P$.
    \end{enumerate}
\end{lemma}
\begin{proof}
    For $(i)$ observe that diagonal prongs are the only pieces of $P$ contained in $(-\infty,0]\times\bR$. Now, for any $\tau\in [-1,M+1]\times[-1,M+1]$ it holds that $Q_v+\tau\in[-2M-1,-M+\eps+1]\times[2M-1,3M+\eps+1]$ and hence $Q_v+\tau\subset (-\infty,0]\times\bR$.

    For $(ii)$ similarly observe that the only pieces of $P$ contained in $\bR\times(-\infty,0]$ are the $x$-prongs of $P$. Further, all $Q_x^{\dindex}$ lie in $\bR\times[-2M,-2M+\eps]$. Hence $Q_x^{\dindex}+\tau$ lies in $\bR\times[-2M,-M+\eps+1]$ and thus can only intersect $x$-prongs of $P$.

    Lastly, for $(iii)$ observe that the only pieces of $P$ contained in $[M+1,\infty)\times\bR$ are the $y$-prongs of $P$. As before, $Q_x^{\eindex}$ lies in $[2M-\eps,2M]\times\bR$ and thus $Q_x^{\eindex}+\tau$ lies in $[2M-\eps-1,3M+1]\times\bR$ and so can only intersect the $y$-prongs of $P$, concluding the proof.
\end{proof}

\begin{lemma}\label{lem:equationInducers}
    Let $\tau=(\tau_x,\tau_y)\in [-1,M+1]\times[-1,M+1]$. Then
    \begin{enumerate}[(i)]
        \item if $\area(P\cap(Q_v+\tau))>0$, then $\exists a_{\aindex}\in A$ such that $\tau_x+\tau_y\in[a_{\aindex}-2\eps,a_{\aindex}+(10n+2)\eps]$,
        \item if $\area(P\cap(Q_x^{\dindex}+\tau))>0$, then $\exists b_{\bindex}\in B$ such that $\tau_x\in [d_{\dindex}+b_{\bindex}+3{\bindex}\eps,d_{\dindex}+b_{\bindex}+(3{\bindex}+2)\eps]$, and
        \item if $\area(P\cap(Q_y^{\eindex}+\tau))$, then $\exists c_{\cindex}\in C$ such that $\tau_y\in [e_{\eindex}+c_{\cindex}+3{\cindex}\eps,e_{\eindex}+c_{\cindex}+(3{\cindex}+2)\eps]$.
    \end{enumerate}
\end{lemma}
\begin{proof}
    For $(i)$ observe that by \Cref{lem:prongprong} $Q_v+\tau$ may only intersect diagonal prongs of $P$. As the ${\aindex}$\ts{th} diagonal prong is defined by (the restriction of) the set $\{(x,y)\in\bR^2\mid a_{\aindex}\leq x+y\leq a_{\aindex}+(10n+2)\eps\}$, and $Q_v+\tau$ is a square with side length $\eps$ with its lower left corner at $(-2M+\tau_x,2M+\tau_y)$ the claim follows.

    Similarly, for $(ii)$, by \Cref{lem:prongprong} $Q_x^{\dindex}+\tau$ may only intersect the $x$-prongs of $P$. As the ${\bindex}$\ts{th} $x$-prong of $P$ is defined by the (restriction of) the set $\{(x,y)\in \bR^2\mid b_{\bindex}+(3{\bindex}+1)\eps\leq x\leq b_{\bindex}+(3{\bindex}+2)\eps\}$ and $Q_x^{\dindex}+\tau$ is a square with side length $\eps$ with its lower left corner at $(-d_{\dindex}+\tau_x,-2M+\tau_y)$, the claim follows.

    The proof for $(iii)$ is analogous to the proof of $(ii)$.
\end{proof}

\begin{lemma}\label{lem:upperboundPerGadget}
    Let $\tau=(\tau_x,\tau_y)\in [-1,M+1]\times[-1,M+1]$. Then
    \begin{enumerate}[(i)]
        \item $\area(P\cap(Q_v+\tau))\leq\eps^2$,
        \item $\area(P\cap(\bigcup_{\dindex} Q_x^{\dindex}+\tau))\leq\eps^2$, and
        \item $\area(P\cap(\bigcup_{\eindex}Q_y^{\eindex}+\tau))\leq\eps^2$.
    \end{enumerate}
\end{lemma}
\begin{proof}
    For $(i)$ observe that $\area(Q_v)\leq\eps^2$.
    
    For $(ii)$, suppose that $\area(P\cap(\bigcup_{\dindex} Q_x^{\dindex}+\tau))>\eps^2$ for the sake of contradiction. Hence, by \Cref{lem:equationInducers}, there are two indices $l\neq l'$ such that there is ${\bindex},{\bindex'}$ such that $\tau_x\in [d_{\dindex}+b_{\bindex}+3{\bindex}\eps,d_{\dindex}+b_{\bindex}+(3{\bindex}+2)\eps]$ and $\tau_x\in [d_{\dindex'}+b_{{\bindex'}}+3{{\bindex'}}\eps,d_{\dindex'}+b_{{\bindex'}}+(3{{\bindex'}}+2)\eps]$. 
    Then for both ${\bindex}$ and ${\bindex'}$ it holds that $3{\bindex}\eps$, $(3{\bindex}+2)\eps$, $3{{\bindex'}}\eps$, and $(3{{\bindex'}}+2)\eps$ are in $[0,5n\eps]=[0,1/(20n)]\subset[0,1)$. Further, both $d_{\dindex}+b_{\bindex}$ and $d_{\dindex'}+b_{{\bindex'}}$ are integers. Thus, $\tau_x\in \mathbb{N} + [3{\bindex}\eps,(3{\bindex}+2)\eps]$ and $\tau_x\in \mathbb{N} + [3{\bindex'}\eps,(3{\bindex'}+2)\eps]$. As both $[3{\bindex}\eps,(3{\bindex}+2)\eps]$ and $[3{\bindex'}\eps,(3{\bindex'}+2)\eps]$ are in $[0,1)$, this is only possible if $[3{\bindex}\eps,(3{\bindex}+2)\eps]$ and $[3{\bindex'}\eps,(3{\bindex'}+2)\eps]$ intersect. As ${\bindex}$ and ${\bindex'}$ are also integers, ${\bindex}={\bindex'}$. But as $\bigcup_{\dindex} Q_x^{\dindex}+\tau\subset \bR\times [-2M+\tau_y,2M+\eps+\tau_y]$ and $P_x^{\bindex}\subset [b_{\bindex}+(3{\bindex}+1)\eps,b_{\bindex}+(3{\bindex}+2)\eps]\times\bR$, in particular 
    \[\area(P\cap(\bigcup_{\dindex} Q_x^{\dindex}+\tau))\leq\area(P_x^{\bindex}\cap(\bigcup_{\dindex} Q_x^{\dindex}+\tau))\leq \eps^2,\]
    concluding the proof for $(ii)$. An analogous argument proves $(iii)$.
\end{proof}

\begin{lemma}\label{thm:reduction}
    Let $P$ and $Q$ be as constructed. The underlying $(3,2)$-SUM instance has a solution if and only if there is a translation $\tau$ such that $\area(P\cap(Q+\tau))\geq 1+3\eps^2$.
\end{lemma}
\begin{proof}
    First assume there is $(a_{\aindex},b_{\bindex},c_{\cindex},d_{\dindex},e_{\eindex})\in A\times B\times C\times D\times E$ such that $a_{\aindex}=b_{\bindex}+c_{\cindex}+d_{\dindex}+e_{\eindex}$. Then let $\tau^*=(d_{\dindex}+b_{\bindex}+(3{\bindex}+1)\eps,e_{\eindex}+c_{\cindex}+(3{\cindex}+1)\eps)$. Clearly $\tau^*\in [0,M]$. Observe that $\area(P_x^{\bindex}\cap(Q_x^{\dindex}+\tau^*))=\eps^2$, $\area(P_y^{\cindex}\cap(Q_y^{\eindex}+\tau^*))=\eps^2$, $\area(P_d^{\aindex}\cap(Q_v+\tau^*))=\eps^2$, and $\area(P_A\cap(Q_A+\tau^*))=1$. Hence, by \Cref{lem:upperboundPerGadget}
    \[\area(P\cap(Q+\tau^*))=1+3\eps^2.\]
    Now let $\area(P\cap(Q+\tau))\geq 1+3\eps^2$ for some $\tau$. Then \Cref{lem:anchor} implies that $\tau\in[-1,M]\times[-1,M]$. But then \Cref{lem:upperboundPerGadget}, the pidgeon hole principle, and the fact that $\area(Q_A)=1$ imply that $\area(P\cap(Q_v+\tau))=\eps^2$, $\area(P\cap(\bigcup_{\dindex} Q_x^{\dindex}+\tau))=\eps^2$, and $\area(P\cap(\bigcup_{\eindex}Q_y^{\eindex}+\tau))=\eps^2$. But then \Cref{lem:equationInducers} implies that there is $(a_{\aindex},b_{\bindex},c_{\cindex},d_{\dindex},e_{\eindex})\in A\times B\times C\times D\times E$ with 
    \begin{enumerate}
        \item $\tau_x\in [d_{\dindex}+b_{\bindex}+3{\bindex}\eps,d_{\dindex}+b_{\bindex}+(3{\bindex}+2)\eps]$,
        \item $\tau_y\in [e_{\eindex}+c_{\cindex}+3{\cindex}\eps,e_{\eindex}+c_{\cindex}+(3{\cindex}+2)\eps]$, and 
        \item $\tau_x+\tau_y\in[a_{\aindex}-2\eps,a_{\aindex}+(10n+2)\eps]$.    
    \end{enumerate}
    This in turn implies that the interval
    \[[d_{\dindex}+b_{\bindex}+e_{\eindex}+c_{\cindex}+3(k+j)\eps,d_{\dindex}+b_{\bindex}+e_{\eindex}+c_{\cindex}+(3(k+j)+4)\eps]\]
    intersects $[a_{\aindex}-2\eps,a_{\aindex}+(10n+2)\eps]$. This implies in particular that
    \begin{align*}
        d_{\dindex}+b_{\bindex}+e_{\eindex}+c_{\cindex}&\in [a_{\aindex}-2\eps+3(k+j)\eps,a_{\aindex}+(10n+2)\eps+(3(k+j)+4)\eps]\\
        &\subset [a_{\aindex}-2\eps,a_{\aindex}+(25n)\eps]=[a_{\aindex}-2\eps,a_{\aindex}+1/(4n)].
    \end{align*}
    As both $2\eps$ and $1/(4n)$ is less than $1$, and $d_{\dindex}+b_{\bindex}+e_{\eindex}+c_{\cindex}$ and $a_{\aindex}$ are integers, in particular $a_{\aindex}=d_{\dindex}+b_{\bindex}+e_{\eindex}+c_{\cindex}$ concluding the proof.
\end{proof}

\thmFiveSum*
\begin{proof}
    This is an immediate consequence of \Cref{thm:reduction} and \Cref{lem:32sum} together with the fact that any coordinate in $P$ and $Q$ as constructed for \Cref{thm:reduction} is polynomial in the underlying $(3,2)$-SUM instance.
\end{proof}

\subsection{Extensions to related problems}

We remark that for $m=1$, a minor modification of our construction implies similar bounds for the polygon containment problem. The minor modifications consist of the following: We remove the anchor gadget $P_A$ for $P$ and place the anchor gadget $Q_A$ at $[0,1]\times[0,1]$. This similarly restricts translations to $[0,M]\times[0,M]$, as otherwise $Q_A$ is not contained in the origin polygon of $P$. We consider the $3$-SUM instance as a $(3,2)$-SUM instance with $D=E=\{0\}$. Note that this modified version of $P$ does not require any connecting pieces. For $Q$, the connecting pieces that connect the unique $x$- and $y$-prongs of $Q$ to the origin are straight lines. Hence, they are contained in any $P_x^{\bindex}$ and $P_y^{\cindex}$ if the corresponding prong of $Q$ is contained in $P_x^{\bindex}$ and $P_y^{\cindex}$ respectively. Lastly, the connecting piece for $Q_v$ is a diagonal line from $Q_v$ to $0$, such that it is contained in the diagonal prong of $P$ if it contains $Q_v$. 

Thus for the problem of \polyCont we obtain a similar lower bound to the presented lower bound for the maximum overlap problem.

\begin{problem}[\polyCont]
    Let $P$ and $Q$ be two polygons in the plane with complexities $n$ and $m$ respectively. Is there a translation of $Q$ such that it is contained in $P$?
\end{problem}


\polyContLowerBound*

We believe that our reduction extends to higher dimensions encoding imbalanced $(2d+1)$-SUM instances as \textsc{PolytopeMaxOverlap} instances in $d$ dimensions where instead of polygons we want to maximize the measure of the overlap of two $d$-dimensional polytopes. Similarly, it encodes $(d+1)$-SUM instances as \textsc{PolytopeContainment} instances in $d$ dimensions, where the goal is to translate a $d$-dimensional polytope such that it is contained in another $d$-dimensional polytope. Assuming the $k$-SUM Hypothesis, for \textsc{PolytopeContainment} in $d$ dimensions this reduction yields a lower bound of $\O(n^{\lceil(d+1)/2\rceil-\eps})$, while for \textsc{PolytopeMaxOverlap} this reduction yields a lower bound of $\O(\max(n^{\lceil(d+1)/2\rceil},nm^d)^{1-\eps})$.


\section{Open Problems}
We leave multiple intriguing open problems for future work:
\begin{itemize}
    \item It remains open what the optimal running time is for polygon containment and overlap on simple polygons when $m \in \omega(\mathrm{polylog}(n))$.
    \item Surprisingly, we still do not know whether polygon overlap is harder than polygon containment for any class of polygons or whether these problems have essentially the same algorithmic complexity.
    \item To tackle the above problem, it would be interesting to resolve whether our algorithm for orthogonal polygons can be improved or show that it is tight with a lower bound for orthogonal polygons.
\end{itemize}

\bibliographystyle{plainurl}
\bibliography{bibliography}

\end{document}